\pdfoutput=1
\RequirePackage{ifpdf}
\ifpdf
\documentclass[pdftex]{sigma}
\else
\documentclass{sigma}
\fi

\usepackage{amscd}
\numberwithin{equation}{section}

\newtheorem{thm}{Theorem}[section]
\newtheorem{Proposition}[thm]{Proposition}
{\theoremstyle{definition}
\newtheorem{Definition}[thm]{Definition}
\newtheorem{Remark}[thm]{Remark}}

\newcommand{\C} {\mathbb{C}}
\newcommand{\Cs} {\mathbb{C}_s}

\newcommand{\R}{\mathbb{R}}

\newcommand{\by}{\boldsymbol{y}}
\newcommand{\bx}{\boldsymbol{x}}
\newcommand{\bt}{\boldsymbol{t}}

\newcommand{\al}{\alpha}

\newcommand{\de}{\delta}
\newcommand{\ep}{\epsilon}

\newcommand{\rSU}{\mathrm{SU}}

\newcommand{\rSL}{\mathrm{SL}}

\newcommand{\rM}{\mathrm{M}}

\newcommand{\tr}{\mathrm{tr}}

\newcommand{\ddet}{{\mathrm{det}}}

\newcommand{\lra} {\longrightarrow}

\begin{document}
\allowdisplaybreaks

\newcommand{\arXivNumber}{1705.01755}

\renewcommand{\PaperNumber}{066}

\FirstPageHeading

\ShortArticleName{Quantum Klein Space and Superspace}

\ArticleName{Quantum Klein Space and Superspace}

\Author{Rita FIORESI~$^\dag$, Emanuele LATINI~$^{\dag\ddag}$ and Alessio MARRANI~$^\S$}

\AuthorNameForHeading{R.~Fioresi, E.~Latini and A.~Marrani}

\Address{$^\dag$~Dipartimento di Matematica, Universit\`{a} di Bologna,\\
\hphantom{$^\dag$}~Piazza di Porta S.~Donato 5, I-40126 Bologna, Italy}
\EmailD{\href{mailto:rita.fioresi@UniBo.it}{rita.fioresi@UniBo.it}, \href{mailto:emanuele.latini@UniBo.it}{emanuele.latini@UniBo.it}}

\Address{$^\ddag$~INFN, Sez. di Bologna, viale Berti Pichat 6/2, 40127 Bologna, Italy}

\Address{$^\S$~Museo Storico della Fisica e Centro Studi e Ricerche ``Enrico Fermi'',\\
\hphantom{$^\S$}~Via Panisperna 89A, I-00184, Roma, Italy}
\EmailD{\href{mailto:jazzphyzz@gmail.com}{jazzphyzz@gmail.com}}

\ArticleDates{Received February 23, 2018, in final form June 15, 2018; Published online June 28, 2018}

\Abstract{We give an algebraic quantization, in the sense of quantum groups, of the complex Minkowski space, and we examine the real forms corresponding to the signatures $(3,1)$, $(2,2)$, $(4,0)$, constructing the corresponding quantum metrics and providing an explicit presentation of the quantized coordinate algebras. In particular, we focus on the Kleinian signature $(2,2)$. The quantizations of the complex and real spaces come together with a coaction of the quantizations of the respective symmetry groups. We also extend such quantizations to the $\mathcal{N}=1$ supersetting.}


\Classification{17B37; 16T20; 20G42; 81R50; 17B60}

\section{Introduction} \label{intro-sec}

Undergoing crucial advances from Euclid to Newton and then from Minkowski to Einstein, classical space-time has been one of the pillars in the conceptual building of physics. The theory of relativity formulated by Einstein stressed the crucial role played by concept of locality, which lies at the heart of the formulation of effective field theory, one of the most powerful descriptive frameworks to capture the features of a variety of physical systems at relatively low-energies.

Thus, it is surprising that quantum physics had little or no influence at all on the concepts of space and time, despite the essential nature of non-locality being one of the striking (and most counter-intuitive, though) features of quantum mechanics. If one excludes quantum gravity, all theories characterized by quantum processes and dynamics ultimately relies on a~classical picture of space-time.

A quantum theory of gravity is essentially non-local, as it immediately follows from the existence of a fundamental length, the Planck length. Indeed, in both superstring/M-theory and loop quantum gravity, a minimal distance measure occurs, hinting for the fact that the infinitely-differentiable space-time manifold may be an illusory picture, breaking down at very small scales. This suggests that space-time may be quantized; Snyder \cite{Snyder} was the first to introduce the term ``quantized space-time''. Consequently, quantum gravity should face the challenging aim of understanding the nature of quantum space-time, because it is reasonable to expect that, at some fundamental level, the classical notion of space-time should be superseded by some suitable quantum notion.

As it seems intuitively natural from the perspective of local effective field theories, the usual approach to such fundamental conceptual issues in contemporary theoretical physics is that quantum space-time should be related to phenomena in which a strong gravitational field is involved, and/or that the quantum modifications of classical space-time are limited to very tiny regions, namely of the size of the Planck length. It thus follows that the fundamental properties of quantum space-time would not be pertaining to ordinary non-gravitational phenomena.

One of the basic ways in which the classical notion of space-time is generalized is non-commutative geometry (cf., e.g., \cite{NC-1, NC-2,frt, NC-3}, and references therein). In this framework, the non-commutative algebra of `quantum space-time coordinates' \cite{NC-1, NC-5, NC-6, ma2, NC-4, NC-3} works as an ultra-violet regulator, with the Planck length being the minimal resolution length scale which generates uncertainty relations among non-commuting coordinates. This would allow to describe the `fuzzy' or `discrete' nature of the space-time, occurring only at very small distances or high energies \cite{NC-7, NC-8}. Even if most formulations of non-commutative space-time suffer from explicit violations of the Lorentz invariance, many examples of covariant non-commutative space-times, retaining all global symmetries of the underlying classical theory, have been considered in literature \cite{Luk-1, Luk-2, Snyder, Yang}.

Studies of simple dynamical models in quantized gravitational background (see, e.g., \cite{quantum-3, quantum-2, quantum-1}) also indicate that the Lie-algebraic space-time (Lorentz, Poincar\'{e}) symmetries are modified into quantum symmetries, described by non-cocommutative Hopf algebras; these latter, after Drinfeld, have been named quantum deformations or quantum groups~\cite{quantum-groups}.

It is here worth recalling that most of the `quantum' space-times investigated so far are non-commutative versions of the Minkowski space-time, thus with signature $(s,t)=(3,1)$ \cite{M-4, fi1, Luk-1, M-2, M-1, M-3}. In \cite{Zumino}, Zumino, Wess, Ogievetsky and Schmidke adopted an \textit{ad hoc} approach for the quantization of the $D=4$ Minkowski space, exploiting the coaction of the Poincar\'{e} group on it. In \cite{cfl2,cfl1, Cerv-1} the quantum deformation of the complex (chiral) Minkowski and conformal superspaces was investigated by exploiting the formal machinery of flag varieties developed in~\cite{fi1, Fior-2}. Recently, in~\cite{Spin}, another space-time quantization based on the spinorial description, namely a~string theory inspired ${\rm Spin}(3,1)$ worldsheet action, was introduced. Moreover, some recent studies considered the construction of non-commutative space-times with non-vanishing cosmological constant, dealing with the intriguing issue of the interplay between quantum gravity effects and the non-vanishing curvature of space-time, with interesting cosmological consequences of Planck scale physics (cf., e.g., \cite{Ballesteros} and references therein).

As evident from above, in the context of the study of non-commutative space-time structures at Planckian distances, the procedure of deformation of the space-time coordinates and of the space-time symmetries plays a key role. In this respect, one of the main approaches to the classification of deformations has been provided by the theory of classical $r$-matrices \cite{r-1, cp, r-4, r-2}; concerning the various possible signatures in $D=4$, recently, in \cite{Lukk, Lukk-2,Lukk-3} the deformations of the orthogonal Lie algebras $\mathfrak{o}(4-k,k)$ (with $k=0,1,2$), were investigated (together with the quaternionic real form $\mathfrak{o}^{\ast }(4)$), thus dealing with all possible real forms of $\mathfrak{o}(4;\mathbb{C})$.

In the present paper, we exploit a novel procedure which allows us to investigate the quantization deformations of the $D=4$ space-time in all possible signatures; in particular, we focus on the less known case with two timelike and two spacelike dimensions, namely on the Kleinian (also named ultrahyperbolic) signature $(s,t)=(2,2)$. The corresponding real form of $\mathfrak{o}(4;\mathbb{C})$ is~$\mathfrak{o}(2,2)$, which is also used in two-dimensional double field theory (see, e.g., \cite{2D-1, 2D-2}) or employed as $D=3$ AdS Lie algebra\footnote{The quantum deformation of the $D=3$ AdS group was constructed in~\cite{Ballesteros1}.}. Geometries in Kleinian signature currently remains a vast and yet unexplored realm, with a rich mathematical structure, investigated only in a~few papers (cf., e.g., \cite{K-1, K-2, K-3, K-4, K-5, K-6}). The study of the Klein signature should not be considered as a mere mathematical \textit{divertissement}, but rather its interest stems from relevant physical motivations. Just to name a few, the study of symmetries of scattering amplitudes in super Yang--Mills theories and in supergravity stressed out the relevance of Kleinian signature, especially in $D=4$; in fact, in~\cite{OV} Ooguri and Vafa showed that~$\mathcal{N}=2$ superstring is characterized by critical dimension $D=4$, with bosonic part given by a~self-dual metric of signature $s=t=2$. Moreover, $4$-dimensional Kleinian signature essentially pertains to twistors~\cite{twistors}, which provide a~powerful computational tool of scattering amplitudes~\cite{ampli}. In~\cite{flm} the $4$-dimensional Klein space~$\mathcal{M}_{2,2}$ was studied, through its definition inside the related Klein-conformal space, along with its supersymmetric extensions, namely the Klein $\mathcal{N}=1$ superspace~$\mathcal{M}_{2,2|1}$ and the corresponding Klein-conformal $\mathcal{N}=1$ superspace. To the best of our knowledge, this, together with our previous work~\cite{flm}, is the first investigation of the superspace with Kleinian (bosonic) signature.

A direct approach to quantum deformations stems from \cite{FL-1, flm}, and it exhibits an intrinsic elegance based on split algebras $\mathbb{A}_{s}$'s in particular the split complex numbers $\mathbb{C}_{s}$. This is the subject of the present paper, which, as mentioned, will consider the quantum deformation of both real Klein and Klein-conformal $\mathcal{N}=1$ spaces. We will not pursue the aforementioned approach based on $r$-matrices, but rather we will deal with the procedure introduced by Manin~\cite{Manin}; it is here worth pointing out that Manin's approach is equivalent to the one based on $r$-matrices, at least for ${\rm SL}(n)$ groups, which for $n=4$ is the case under investigation. In the 1980's, developing ideas in part due to Penrose a few years earlier~\cite{Penrose}, he introduced the approach to $D=4$ Minkowski space as the manifold of the points of a big cell in the Grassmannian of complex two-dimensional subspaces of a complex four dimensional space (twistor space), extending this framework to the supersymmetric ($\mathcal{N}=1$) case~\cite{ma1}. After some general but rather abstract studies in~\cite{RL,TT}, the study of quantization of flag manifolds and of complex Minkowski space \textit{\`{a}~la Manin} was worked out in~\cite{fi1}, in which were also given the two involutions of the quantum complex Minkowski space, respectively defining the real Minkowski space-time and the real Euclidean space in four dimensions.

The plan of the paper is as follows.

In Section~\ref{cs-sec} we give a brief account of the classical Minkowski and Klein spaces, both in the complex and real setting, together with their symmetry groups.

In Section~\ref{prelim-sec}, we consider quantum groups and their homogenous spaces. We approach the theory of deformation according to Manin~\cite{ma2}, that is, we quantize the coordinate rings of our algebraic geometric objects, spaces and groups. This approach is equivalent to the $r$-matrix formulation, proposed originally by Fadeev et al.\ in~\cite{frt} (see also~\cite{cp} for a complete account, and references therein). Our approach is slightly more general than the ones present the literature, since we take our ground ring~$A$ to be a $k$-algebra, where~$k$ is~$\R$, $\C$ or~$\C_s$. Observe that one could also furtherly generalize this approach by substituting the complex and split complex units $i$ and $j$ with the so called dual or parabolic one $\epsilon$, with $\epsilon^2=0$. The use of the parabolic unit plays a relevant role within the context of non-Euclidean spaces \cite{epsilon} and it was applied in the framework of quantum groups and quantum Cayley--Klein algebras in \cite{Ballesteros2,Ballesteros3, Gromov}.

As we shall see, this allows a unified treatment for all signatures of the real forms of the complex Minkowski space, and provides a natural generalization to the quantum supersetting. Furthermore, in our treatment, the actions of the symmetry groups get quantized \textit{in tandem} with the homogeneous spaces, though we have to interpret them as \textit{coactions}, since, in the quantum group setting, as usual, the geometric space is replaced by its algebra of functions (in this very case, the algebra of polynomial functions).

In Section~\ref{real-sec} we examine the real forms of the quantum spaces introduced in the previous section, thus producing a deformation of Euclidean, Minkowski and Klein spaces, together with their metrics. We also compute explicitly the commutation relations among generators for the quantum rings of the Euclidean, Minkowski and Klein spaces, thus giving a presentation of such quantum rings. The metrics appear naturally as the quantum determinants of suitable quantum matrices. Our approach allows to obtain at one stroke also the quantization of the symmetry groups: we in fact obtain also the quantum Poincar\'{e} groups as well as a coaction of such quantum groups on the corresponding quantum spaces.

Finally, in Section~\ref{smink-sec} we generalize the constructions of the previous sections to the $\mathcal{N}=1$ supersetting. In particular, in Section~\ref{qsgrps-subsec} we introduce the notion of quantum supergroup, whereas in Section~\ref{qsmink-subsec} we consider the quantum chiral Minkowski superspace. Then, for brevity's sake, in Section~\ref{superreal-subsec} we focus (working on $\mathbb{C}_{s}$) on the real form associated to the Klein involution.

\section{The classical spaces} \label{cs-sec}
In this section we describe the classical Minkowski and Klein spaces using an approach, which turns out particular fruitful for our subsequent treatment of the quantizations of these spaces together with their symmetry groups and real forms.

\subsection{The (split) complex Minkowski space} \label{cm-subsec}

The complex Minkowski space can be realized inside the conformal space $\rSL_4(\C)/P$, $P$ being an upper maximal parabolic, as an open set, called the {\it big cell}. It admits a natural action of the Poincar\'e group and Weyl dilations, which sit as subgroups inside the conformal group (we address the reader to~\cite{fl,flv} for more details) that reads as
\begin{gather}
\underbrace{\begin{pmatrix}x&0\\tx&y\end{pmatrix}}_{\textrm{Poincar\'e$\times$dilations}},\underbrace{\begin{pmatrix} I_2 \\ n \end{pmatrix}}_{\textrm{big cell}}\mapsto
\underbrace{ \begin{pmatrix} I_2 \\ ynx^{-1}+t \end{pmatrix}}_{\textrm{big cell}} ,\label{poincare}
\end{gather}
where all of $x$, $y$, $t$ and $n$ are $2 \times 2$ matrices with complex coefficients. Notice that $\mathrm{diag}(x,y)$ parametrizes dilations, while translations are represented by the matrix $t$. We then identify
the complex Minkowski space with the translational part of the Poincar\'{e} group, that is the space of $2\times 2$ complex matrices. Note that the subgroup
\begin{gather*}
H=\left\{\begin{pmatrix}x&0\\0&y\end{pmatrix}, x,y\in \rSL_2(\C) \right\}
\end{gather*}
preserve the determinant of $n$, that we then identify with the norm of a vector; this establishes the homomorphism
\begin{gather*}
\rSL_2(\C)\times \rSL_2(\C)\rightarrow \textrm{SO}_4(\C).
\end{gather*}
In \cite{flm} we have shown how this picture can be consistently generalized by substituting $\C$ with $\Cs$ in order to consider a more general picture and different real forms; we will refer to it generally as the complex Minkowski space $\mathcal{M}^k_{4}$, and it will be clear from the context which algebra we use; the determinant of the various real forms will naturally select different pseudo-Riemannian signatures.

\subsection{Real forms and involutions} \label{rfi-subsec}

Let $k=\C$ or $\Cs$ and $\mathcal{A}$ be a commutative algebra over $k$. An involution is a map $*\colon \mathcal{A} \lra \mathcal{A}$ satisfying the properties
\begin{gather*}
(\alpha f+\beta g)^*=\bar{\alpha}f^*+\bar{\beta}g^*,\qquad
(fg)^*=f^*g^*,\qquad
(f^*)^*=f,
\end{gather*}
where $f,g\in \mathcal{A}$, $\alpha,\beta\in k$ and the bar indicates the usual conjugation in $k$. In this setting, we do not need to specify whether the involution is multiplicative or antimultiplicative, since we are working with commutative algebras, while in the quantum case this choice will be crucial. We then consider the following three involutions on functions on $\mathbb{C}[\mathrm{M}_2(k)]$, the algebra of polynomials on the matrices:
\begin{gather*}
*_M\colon \ \begin{pmatrix} a & b \\ c & d \end{pmatrix} \mapsto \begin{pmatrix} a& c \\ b & d\ \end{pmatrix},\\
*_E \colon \ \begin{pmatrix} a & b \\ c & d \end{pmatrix} \mapsto \begin{pmatrix} d & -c \\ -b & a \end{pmatrix},\\
*_K \colon \ \begin{pmatrix} a & b \\ c & d \end{pmatrix} \mapsto \begin{pmatrix} a & b \\ c & d \end{pmatrix}.
\end{gather*}
We refer to $*_M$ (resp.~$*_E$, $*_K$) as the \textit{Minkowski $($resp.\ Euclidean, Klein$)$ involution}. In fact, that, if $k=\C$, the fixed points of these involutions are explicitly given as follows.
\begin{itemize}\itemsep=0pt
\item The fixed points of $*_M$ are Hermitian matrices\footnote{This is the quadratic Jordan algebra over $\mathbb{C}$, and this justifies our notation}; we denote the space of hermitian matrices by $J_{2}^{\mathbb{C}}$ and we write a generic element $m$ explicitly as
\begin{gather*}
m=\begin{pmatrix} x_0+x_1 & x_2+ix_3 \\ x_2-ix_3 & x_0-x_1\end{pmatrix}.
\end{gather*}
Given $m,p\in J_{2}^{\mathbb{C}}$, one can define the bilinear form
\begin{gather*}
g(m,p)=\tfrac 1 2 \tr\big( m \left(p-\tr(p)\right)\big),\qquad Q:=g(m,m)=-\det m
\end{gather*}
with signature $(3,1)$. One can then identify Hermitian matrices with the real Minkowski space $\mathcal{M}_{3,1}\sim \mathbb{R}^{3,1}$ that is naturally equipped with the action of the Lorentz group, $H=\rSL_2(\C)\times \rSL_2(\C)$, sitting as a subgroup into the conformal group. $H$ acts on the Minkowski space by
\begin{gather*}
n\rightarrow y n x^{-1}.
\end{gather*}
Requiring that this action maps $J_{2}^{\mathbb{C}}$ into itself forces $x^{-1}=y^\dag$, thus we recover the canonical action of $\rSL_2(\C)$ on Hermitian matrices, establishing the isomorphism $\rSL_2(\C)\sim\textrm{Spin}(3,1)$.

\item The fixed points of $*_E$ are complex matrices $e\in \rM_2(\C)$ satisfying the relations $e^{\dag}=\omega e^t \omega^{-1}$, with $\omega$ being the usual symplectic matrix. We denote this space by $N_{2}^{\mathbb{C}}$, and the generic element will be decomposed as
\begin{gather*}
e=\begin{pmatrix} x_0+ix_1 & x_2+ix_3 \\ -x_2+ix_3 & x_0-ix_1\end{pmatrix}.
\end{gather*}
On the space of those matrices we define the bilinear form
\begin{gather*}
g(e,f)=\tfrac 1 2 \tr \big(e f^{\dag}\big),\qquad Q:=g(e,e)=\det e
\end{gather*}
that yields the Riemannian metric with signature $(4,0) $. We identify the fixed points of~$*_E$ with the real Euclidean space $\mathcal{M}_{4,0}\sim \mathbb{R}^{4,0}$. The isometry group can be again obtained by looking at the action of the subgroup $H$ and requiring it maps ${N}_{2}^{\mathbb{C}}$ into itself. A~tedious calculation shows that in this case one gets $x,y\in \rSU_2$ as one would expect since $\rSU_2\times \rSU_2\sim \textrm{Spin}(4) $.

\item The fixed points of $*_K$ are real matrices, and the generic element $s \in \rM_2(\R)$ will be written, introducing a pair of light cone coordinates, as
\begin{gather*}
s=\begin{pmatrix} x_0+x_1 & x_2+x_3 \\ x_2-x_3 & x_0-x_1\end{pmatrix}.
\end{gather*}
We define the bilinear form
\begin{gather*}
g(s,p)=\tfrac 1 2 \tr \big(s \omega p^{t} \omega^{-1}\big), \qquad Q:= g(s,s)=\det s .
\end{gather*}
In this case, we obtain a pseudo-Riemannian metric with signature $(2,2)$. The isometry group can be again obtained by forcing $H$ to map real matrices into real matrices. Explicitly, one gets that $x$ and $y$ must be real matrices and thus elements of~$\rSL_2(\R)$.

We can easily see that the described action preserves the bilinear form defined above, and this is again not surprising since $\rSL_2(\R)\times \rSL_2(\R)\sim \textrm{Spin}(2,2)$. The proof relies on the fact that for
$g\in \rSL_2(\R)$ it holds that $g^{-1}=\omega g^t \omega^{-1}$.
\end{itemize}

Note that there is a fourth real, non-compact form of $\mathfrak{o}(4;\mathbb{C)}$: the so-called quaternionic one, $(\mathfrak{s)o}^{\ast }(4)\sim \mathfrak{sl}(2,\mathbb{R})\oplus \mathfrak{su}(2)$; it has been recently
treated in \cite{Lukk, Lukk-2} (see also, e.g.,~\cite{Girelli} for recent applications).

If $k=\C_s$, the fixed points of all the three given involutions yield the Klein space, as one can naively check by substituting $j$ to $i$, where $j$ is the imaginary unit of the split complex num\-bers~$\C_s$ ($j^2=1$). At the classical level, one thus obtains four \textit{equivalent} realizations of the $D=4$ Klein space $\mathcal{M}_{2,2}$. Note that, along the way, one proves also the isomorphisms
\begin{gather*}
\rSL_2(\C_s)=\widetilde{\rSU_2}\times \widetilde{\rSU_2}=\textrm{Spin}(2,2),
\end{gather*}
with $\widetilde{\rSU_2}$ being the unitary $2\times 2$ matrices over $\C_s$ with determinant one.

All this is summarized in the following table:
\begin{center}
\begin{tabular}{|c|c|lll|}
\hline
involution&fixed Points over $k=\C,\C_s$&&space&isometry group\\ \hline
&&$\C$ &$\mathcal{M}_{3,1}$&$\rSL_2(\C)$\\
$*_M$&$J_2^k\colon \big\{m\in \rM_2(k)$ s.t.\ $m=m^{\dag}\big\}$&&&\\
&&$\C_s$ &$\mathcal{M}_{2,2}$&$\rSL_2(\C_s)$ \\[1mm] \hline
&&$\C$ &$\mathcal{M}_{4,0}$&$\rSU_2\times \rSU_2$\\
$*_E$& $N_2^k\colon \big\{e\in \rM_2(k)$ s.t.\ $e^{\dag}=\omega e^t \omega^{-1}\big\}$&&&\\
&&$\C_s$ &$\mathcal{M}_{2,2}$&$\widetilde{\rSU_2}\times \widetilde{\rSU_2}$ \\[1mm] \hline
&&$\C$ &$\mathcal{M}_{2,2}$&$\rSL_2(\R)\times \rSL_2(\R)$\\
$*_K$& $\rM_2(\R)$&&&\\
&&$\C_s$ &$\mathcal{M}_{2,2}$&$\rSL_2(\R)\times \rSL_2(\R)$ \\[1mm] \hline
\end{tabular}
\end{center}

\section{Quantum deformations}\label{prelim-sec}

In this section we introduce the quantum deformations of the coordinate rings of the spaces we have studied above. We construct such deformations together with the coactions of the deformations of their symmetry groups, and we also give quantizations of the real forms, which are compatible with such coactions.

\subsection{Quantum groups} \label{qg-subsec}

We now introduce the quantum special linear group, as well as the notion of quantum space.

Let $k=\C,\C_s$ and $k_q=A\big[q,q^{-1}\big]$, where $q$ is an indeterminate. In applications to quantum physics and quantum gravity, the non-zero real parameter $q$ can be thought as related to the Planck constant as $q\sim e^h$; the classical limit is achieved for $h \lra 0$ and $q\lra 1$.

\begin{Definition} \label{qspace-def}
We define \textit{quantum space} $k_q^n$ the non commutative algebra
\begin{gather*}
k_q^n := k_q \langle x_1, \dots , x_n \rangle/\big(x_ix_j - q^{-1} x_ix_j,\, i<j\big),
\end{gather*}
where $k_q \langle x_1, \dots , x_n \rangle$ is the free algebra over the ring $k_q$ with generators $x_1, \dots , x_n$.

We define the \textit{quantum matrix bialgebra} as
\begin{gather*}
\rM_q(n) := k_q \langle a_{ij} \rangle/I_M,
\end{gather*}
where the indeterminates $a_{ij}$ satisfy Manin's relations \cite{ma2} generating the ideal~$I_M$
\begin{alignat}{3}
& a_{ij}a_{kj}=q^{-1}a_{kj}a_{ij}, \quad i<k \qquad && a_{ij}a_{kl}=a_{kl}a_{ij},\quad i<k,j>l \quad \text{or} \quad i>k,j<l,& \nonumber\\
& a_{ij}a_{il}=q^{-1}a_{il}a_{ij}, \quad j<l, \qquad && a_{ij}a_{kl}-a_{kl}a_{ij}=\big(q^{-1}-q\big)a_{ik}a_{jl}, \quad i<k,j<l,& \label{manin-rel2}
\end{alignat}
 with comultiplication and counit given by
\begin{gather*}
\Delta(a_{ij})=\sum_k a_{ik} \otimes a_{kj},\qquad \ep(a_{ij})=\de_{ij}.
\end{gather*}
\end{Definition}

The definitions of the comultiplication $\Delta$ and the counit $\ep$ are summarized with the following notation
\begin{gather*}
\Delta \begin{pmatrix} a_{11} & \dots & a_{1n} \\
\vdots & & \vdots \\ a_{n1} & \dots & a_{nn} \end{pmatrix}
 =
\begin{pmatrix} a_{11} & \dots & a_{1n} \\
\vdots & & \vdots \\ a_{n1} & \dots & a_{nn} \end{pmatrix} \otimes
\begin{pmatrix} a_{11} & \dots & a_{1n} \\
\vdots & & \vdots \\ a_{n1} & \dots & a_{nn} \end{pmatrix},\\
\ep\begin{pmatrix} a_{11} & \dots & a_{1n} \\
\vdots & & \vdots \\ a_{n1} & \dots & a_{nn} \end{pmatrix} =
\begin{pmatrix} 1 & \dots & 0 \\
\vdots & & \vdots \\ 0 & \dots & 1 \end{pmatrix},
\end{gather*}
$k_q^n$ and $\rM_q(n)$ are {\it quantizations} of the module $k^n$ and the matrix algebra $\rM_n(k)$ respectively, that is, when $q=1$, $k_q^n$ becomes the algebra of polynomials with coefficients in $k$ and similarly holds for $\rM_n(k)$. Since $\rM_n(k)$ has a natural action on $k^n$, we have that dually $\rM_q(n)$ coacts on~$k_q^n$ (see \cite{ma1}, and \cite[Chapter~5]{fl} for more details).

\begin{Proposition}\label{qmatrices-coaction}
We have natural left and right coactions of the bialgebra of quantum matrices on the quantum space $k_q^n$
\begin{alignat*}{3}
& \lambda\colon \ k_q^n \lra \rM_q(n) \otimes k_q^n, \qquad && \rho\colon \ k_q^n \lra k_q^n \otimes \rM_q(n), & \\
& \hphantom{\lambda\colon}{} \ x_i \mapsto \sum_j a_{ij} \otimes x_j,\qquad && \hphantom{\rho\colon}{} \ x_i \mapsto \sum_j x_j \otimes a_{ji}. &
\end{alignat*}
\end{Proposition}

Now we turn to the special linear group $\rSL_n(k)$. The algebra of polynomials with coefficients in $k$ on the special linear group is obtained from the matrix bialgebra by imposing the condition that the determinant is equal to $1$. We consider the \textit{quantum exterior algebra in $n$ variables}
\begin{gather*}
\wedge^n_q := k_q \langle \chi_1, \dots , \chi_n \rangle/\big(\chi_i\chi_j + q \chi_j \chi_i, \chi_i^2, \, i<j\big).
\end{gather*}
Also $\wedge^n_q$ admits a (right) coaction of $\rM_q(n)$, using the transpose as one expects
\begin{gather*}
r\colon \ \wedge^n_q \lra \wedge^n_q\otimes \rM_q(n), \qquad \chi_i \mapsto \sum_j a_{ij} \otimes \chi_j.
\end{gather*}

\begin{Definition} \label{qdet-def}
We define \textit{quantum determinant}, the element in $\rM_q(n)$ $\ddet_q(a_{ij})$ (or $\ddet_q$ for short) defined by the equation
\begin{gather*}
r(\chi_1 \cdots \chi_n)=\ddet_q(a_{ij}) \otimes \chi_1 \cdots \chi_n.
\end{gather*}
\end{Definition}

A small calculation gives
\begin{gather*}
\ddet_q(a_{ij})=\sum_\sigma (-q)^{-\ell(\sigma)}a_{1 \sigma(1)} \cdots a_{n \sigma(n)} =\sum_\sigma (-q)^{-\ell(\sigma)}a_{\sigma(1)1} \cdots a_{\sigma(n)n},
\end{gather*}
where $\sigma$ runs through all the permutations of the first $n$ integers.

Furthermore $\det_q$ is central, i.e., commutes with all the elements in $\rM_q(n)$ and it is a group-like element, that is
\begin{gather*}
\Delta(\ddet_q)=\ddet_q \otimes \ddet_q.
\end{gather*}

\begin{Definition}
We define \textit{quantum special linear group} over $k$ the algebra
\begin{gather*}
\rSL_q(n)= \rM_q(n)/(\ddet_q-1).
\end{gather*}
\end{Definition}
Notice that, most immediately, $\rSL_q(n)$ is a quantum deformation of the special linear group $\rSL_n(k)$.

$\rSL_q(n)$ is an Hopf algebra, with $\Delta$, $\epsilon$ inherited by $\rM_q(n)$ and antipode given by
\begin{gather*}
S_q(a_{ij})=(-q)^{j-i}A_{ji},
\end{gather*}
where $A_{ji}$ is the quantum determinant in the quantum matrix bialgebra generated by the indeterminates $a_{rs}$ with $r,s =1 , \dots , n$ and $r \neq j$, $s \neq i$.

The following definition will be useful later.

\begin{Definition}We say that a certain set of indeterminates $\{a_{ij}\}$ is a \textit{quantum matrix} if the~$a_{ij}$'s satisfy Manin's relations~(\ref{manin-rel2}). We also speak of a~{\it $q^{-1}$-quantum matrix}, meaning that it satisfies Manin's relations~(\ref{manin-rel2}) with $q$ replaced by $q^{-1}$.
\end{Definition}

\subsection{The quantum spaces}\label{qkm-subsec}

We want to quantize the setting introduced in Sections~\ref{cs-sec} and~\ref{cm-subsec}. We proceed heuristically imposing the following equality (see also \cite[Chapter~5]{fl})
\begin{gather} \label{eur-der}
\begin{pmatrix} x & 0 \\ tx & y \end{pmatrix}
\begin{pmatrix}I_2 & s \\ 0 & I_2\end{pmatrix}
=\begin{pmatrix}a_{11} & \dots & a_{14} \\
 \vdots & & \vdots \\
 a_{41} & \dots & a_{44}\end{pmatrix},
\end{gather}
which will yield the variables $x$, $y$, $t$ in terms of the coordinates $a_{ij}$ of $\rSL_q(n)$ (for the physical meaning of $x$, $y$, $t$ refer to Section~\ref{cm-subsec}). Equation~(\ref{eur-der}) holds only when $D_{12}^{12}$ is invertible. After a small calculation we get the equalities
\begin{gather*}
x=(x_{ij})=\begin{pmatrix}a_{11} & a_{12} \\ a_{21} & a_{22}\end{pmatrix} , \qquad
t=(\bt_{kj})=\begin{pmatrix}-q^{-1}D_{23} {D_{12} }^{-1} &
D_{13} {D_{12} }^{-1} \\
-q^{-1}D_{24} {D_{12} }^{-1} &
D_{14} {D_{12} }^{-1} \end{pmatrix},\\
s=\begin{pmatrix}-q{D_{12} }^{-1}D_{12}^{23} &
-q{D_{12} }^{-1}D_{12}^{24} \\
{D_{12} }^{-1}D_{12}^{13} &
{D_{12} }^{-1}D_{12}^{14} \end{pmatrix},\qquad
y=(y_{kl})=\begin{pmatrix}D_{123}^{123}{D_{12} }^{-1} &
D_{124}^{123}{D_{12} }^{-1} \\
D_{123}^{124}{D_{12} }^{-1} &
D_{124}^{124}{D_{12} }^{-1} \end{pmatrix},
\end{gather*}
where $1 \leq i,j \leq 2$, $3 \leq k,l \leq 4$ and $D^{j_1 \dots j_r}_{i_1 \dots i_r}$ is the quantum determinant obtained by taking the rows $i_1, \dots , i_r$ and columns $j_1, \dots ,j_r$ (we may omit the column index, when we take the first $r$ columns).

The entries of the matrix $t$ correspond to the generators of the quantum Minkowski space: in fact, we identify the Minkowski space with the translational part of the Poincar\'{e} group as it is classically expressed in formula~(\ref{poincare}). For the time being, we write the elements of $t$, $x$, $y$, $s$ in a matrix form, for convenience. We observe that $t$ is not a quantum matrix, but it is close to it; we in fact obtain a quantum matrix by exchanging the two columns, i.e.,
$\left(\begin{smallmatrix}\bt_{32} & \bt_{31} \\ \bt_{42} & \bt_{41}\end{smallmatrix}\right)$ is a quantum matrix. In details we have
\begin{gather*}
\bt_{32}\bt_{31} = q^{-1} \bt_{31}\bt_{32},\qquad \bt_{32}\bt_{42} = q^{-1}\bt_{42}\bt_{32},\qquad \bt_{31}\bt_{41} = q^{-1}\bt_{41}\bt_{31},\qquad
\bt_{42}\bt_{41} = q^{-1} \bt_{41}\bt_{42},\\
\bt_{32}\bt_{41} = \bt_{41}\bt_{32}+\big(q^{-1}-q\big)\bt_{31}\bt_{42},\qquad
\bt_{31}\bt_{42} = \bt_{42}\bt_{31}.
\end{gather*}

This heuristic reasoning leads to the following definition in analogy with the classical setting.

\begin{Definition} We define \textit{quantum Poincar\'{e} group} $\mathcal{P}_{q}$ as the algebra generated inside $\rSL_q(4)$ by the elements in $x$, $y$, $t$ described above. This is a quantum deformation of the Poincar\'{e} group. We define \textit{quantum Minkowski space} the subring $\mathcal{M}_{4}^{k,q}$ in $\rSL_q(4)$ generated by the elements $\bt_{kj}$. This is a quantum deformation of the coordinate algebra of the translations.

We define the \textit{quantum Lorentz group} as
\begin{gather*}
\mathcal{L}_{q}=\mathcal{P}_{q}/(t, {\det}_q(x)-1, {\det}_q(y)-1).
\end{gather*}
\end{Definition}

Our definition is very natural and in fact gives us a coaction of the quantum Poincar\'{e} group on the Minkowski space.

\begin{Proposition} \label{coaction-min}The quantum Minkowski space admits a natural coaction of the quantum Poincar\'{e} group
\begin{align*}
\de\colon \ \mathcal{M}_{4}^{k,q} & \longrightarrow \mathcal{P}_{q} \otimes \mathcal{M}_{4}^{k,q}, \nonumber \\
 \bt_{ij} & \longmapsto \sum_{s,r} y_{is}S(x_{rj}) \otimes \bt_{sr} +t_{ij} \otimes 1,
\end{align*}
where we rescale all indices $i,j,r,s=1,2$.
\end{Proposition}

\begin{proof} This is a natural consequence of our construction. In fact the coaction $\delta$ corresponds to the restriction of the comultiplication $\Delta$ of $\rSL_q(4)$. In particular we have
\begin{gather*}
\Delta \begin{pmatrix} x & 0 \\ tx & y \end{pmatrix}=\begin{pmatrix} x & 0 \\ tx & y \end{pmatrix}\otimes \begin{pmatrix} x & 0 \\ tx & y \end{pmatrix},
\end{gather*}
implying
\begin{gather*}
\Delta(tx)=tx\otimes x+y\otimes tx.
\end{gather*}
Multiplying then by $S(x)\otimes S(x)$ on both sides one gets the result. Notice that in our notation we identify the translation matrix $t_{ij}$ with the coordinates on the Minkowski space $\bt_{ij}$.
\end{proof}

This gives also a coaction of the quantum Lorentz group by setting the generators $t=0$.

We have then that the quantum Minkowski space is a quantum homogeneous space for the quantum Poincar\'e group and the quantum Lorentz group see also \cite{fl} for more details on quantum homogeneous spaces in general.

\begin{Remark} We consider two quantum planes $\C_q[\chi_1,\chi_2]$ and $\C_{q}[\psi_1,\psi_2]$, that is $\chi_2\chi_1=q \chi_1\chi_2$ and $\psi_2\psi_1=q \psi_1\psi_2$,
then one has the following morphism
\begin{align*}
\C_q^2\stackrel{\ell}{\oplus}
\C_q^2 & \lra \mathcal{M}^{\ell,q}_{\mathbb{C}}, \\
\begin{pmatrix} \chi_1\\\chi_2 \end{pmatrix}\otimes \begin{pmatrix}
\psi_2 & -q\psi_1 \end{pmatrix}
& \mapsto \begin{pmatrix} t_{11} & t_{12} \\
t_{21} & t_{22}
\end{pmatrix}.
\end{align*}
Thus, in analogy with the classical case, we can view the quantum plane as quantum spinors.
\end{Remark}

\section{Real forms of the quantum Minkowski space}\label{real-sec}
In this section we want to construct real forms of the complex quantum Minkowski space $\mathcal{M}_{4}^{k,q}$ as quantum homogeneous space for the quantum Lorentz group. In other words, we want real forms, which are compatible with the coaction $\de$ of Proposition~\ref{coaction-min}, setting to zero the translation part in~$\mathcal{P}_{q}$.
As usual $k=\C$ or $\C_s$ and we shall specify which case, depending on the real form under consideration.

\subsection{Quantum real forms and involutions}
In the quantum setting, real forms correspond to involution of quantum algebras. In particular, if $G_q$ is a quantum group, a real form of $G_q$ is a pair $(G_q,*_q)$, where $*_q$ is an antilinear, involutive, antimultiplicative map respecting the comultiplication and the antipode (see, e.g., \cite[Section~5.3]{fl} for more details). For a quantum homogeneous space, we give the following definition of real form.

\begin{Definition}
Let $V_q$ be a quantum homogeneous space for the quantum group $G_q$, with coaction $\de_{V_q}\colon V_q\lra G_q \otimes V_q$. Assume $(G_q,*_{G_q})$ is a~real form of $G_q$. We say that $(V_q,*_{V_q})$ is a~\textit{real form as quantum homogeneous space} of $V_q$ if

\begin{itemize}\itemsep=0pt
\item $*_{V_q}\colon V_q \lra V_q$ is an involutive antilinear map,
\item $ *_{V_q}$ is antimultiplicative, that is $(ab)^{*_{V_q}}=b^{*_{V_q}}a^{*_{V_q}}$, or multiplicative $(ab)^{*_{V_q}}=a^{*_{V_q}}b^{*_{V_q}}$,
\item $*_{V_q}$ preserves the coaction, that is
\begin{gather*}
\de_{V_q}(a^{*_{V_q}})=\de_{V_q}(a)^{*_{G_q} \times *_{V_q}}.
\end{gather*}
\end{itemize}
\end{Definition}

\begin{Remark}In the following, the compatibility condition with the coaction of the real form of~$G_q$ discussed above, will force us to consider also involution for $G_q$ and $V_q$ that are not antimultiplicative\footnote{In literature, multiplicative involution for quantum groups have been already considered, see for example~\cite{BGST}.}; this is not crucial, since we are not interested in studying particular
representations of the real form of $G_q$, but we are interested just on its action on $V_q$.
\end{Remark}
Inspired by the classical case (treated in Section~\ref{cm-subsec}), we now consider three different involutions on $\mathcal{M}_{4}^{k,q}$. As we shall see later, they are compatible with a suitable real form of the complex
Poincar\'{e} quantum group, thus consistently realizing real forms of~$\mathcal{M}_{4}^{k,q}$ as quantum homogeneous space for the suitable quantum group of symmetries:{\samepage
\begin{align*}
\bullet \ *_{_{Mq}}\colon \ &
\mathcal{M}_{4}^{k,q}  \longrightarrow \mathcal{M}_{4}^{k,q}, \\
& \begin{pmatrix}\bt_{31} & \bt_{32} \\ \bt_{41} & \bt_{42} \end{pmatrix}
  \longmapsto \begin{pmatrix}\bt_{31} & \bt_{41} \\ \bt_{32} & \bt_{42} \end{pmatrix},\\
& q \longmapsto q;\\
\bullet \ *_{_{Eq}} \colon \ &
 \mathcal{M}_{4}^{k,q}  \longrightarrow \mathcal{M}_{4}^{k,q}, \\
& \begin{pmatrix}\bt_{31} & \bt_{32} \\ \bt_{41} & \bt_{42} \end{pmatrix}
\longmapsto \begin{pmatrix}-q^{-1}\bt_{42} & \bt_{41} \\\bt_{32} & -q\bt_{31} \end{pmatrix},\\
& q \longmapsto q;\\
\bullet \ *_{_{Kq}}\colon \ &
\mathcal{M}_{4}^{k,q}  \longrightarrow \mathcal{M}_{4}^{k,q}, \\
& \begin{pmatrix}\bt_{31} & \bt_{32} \\ \bt_{41} & \bt_{42} \end{pmatrix}
\longmapsto \begin{pmatrix}\bt_{31} & \bt_{32} \\ \bt_{41} & \bt_{42}
\end{pmatrix},\\
& q \longmapsto q.
\end{align*}}

Notice that while $*_{Mq}$ and $*_{Eq}$ are antimultiplicative involutions of $k_q[\bt_{ij}]/I_M$ (with $I_M$ denoting the ideal of Manin's relations~(\ref{manin-rel2})) fixing $k_q$, $*_{Kq}$ gives a multiplicative involution of $k_q$; the antimultiplicative property could be recovered sending $q$ to $q^{-1}$. Let us stress once more that $\left(\begin{smallmatrix}{\bt}_{31} & {\bt}_{32} \\ {\bt}_{41} & {\bt}_{42} \end{smallmatrix}\right)$ is not a quantum matrix (it becomes such by interchanging the first and the second columns). This has clearly a consequence, when computing the commutation relations and the quantum determinant, however we shall pay attention and proceed nevertheless with this notation. We now proceed analyzing case by case the quantum structure for both~$k=\C$ or~$\C_s$.

\subsection{The quantum real Minkowski space}\label{qmin-subsec}

We consider now fixed points of $\mathcal{M}_{4}^{\mathbb{C},q}$ with respect to $*_{_{Mq}}$, i.e., the real form $\big(\mathcal{M}_{4}^{\mathbb{C},q},*_{_{Mq}}\big)$. In analogy with the classical case, we parametrize it by the following ``real variables''
\begin{gather*}
\widetilde{\bx}_{0}=\tfrac 12(\bt_{31}+\bt_{42}),\qquad \widetilde{\bx}_2=\tfrac 12(\bt_{32}+\bt_{41}),\qquad
\widetilde{\bx}_3= \tfrac i {2} (\bt_{41}-\bt_{32}),\qquad \widetilde{\bx}_1=\tfrac 12(\bt_{31}-\bt_{42}), \end{gather*}
that are the following fixed points with respect the action of $*_{_{Mq}}$; their commutation relations are listed below
\begin{gather*}
\widetilde{\bx}_2\widetilde{\bx}_0 = q_+\widetilde{\bx}_0\widetilde{\bx}_2+iq_-\widetilde{\bx}_0\widetilde{\bx}_3,\\
\widetilde{\bx}_2\widetilde{\bx}_1 = q_+\widetilde{\bx}_1\widetilde{\bx}_2+iq_-\widetilde{\bx}_1\widetilde{\bx}_3,\\
\widetilde{\bx}_3\widetilde{\bx}_0 = q_+\widetilde{\bx}_0\widetilde{\bx}_3-iq_-\widetilde{\bx}_0\widetilde{\bx}_2,\\
\widetilde{\bx}_3\widetilde{\bx}_1 = q_+\widetilde{\bx}_1\widetilde{\bx}_3-iq_-\widetilde{\bx}_1\widetilde{\bx}_2,\\
\widetilde{\bx}_0\widetilde{\bx}_1 = \widetilde{\bx}_1\widetilde{\bx}_0,\\
\widetilde{\bx}_2\widetilde{\bx}_3 = \widetilde{\bx}_3\widetilde{\bx}_2+iq_-\big(\widetilde{\bx}_0^2-\widetilde{\bx}_1^2\big),
\end{gather*}
with $q_+=\frac 1 2\big(q^{-1}+q\big)$ and $q_-=\frac 1 2\big(q^{-1}-q\big)$. The metric is given by
\begin{gather*}
Q_q=\det\nolimits_q\begin{pmatrix}\bt_{32}&\bt_{31}\\\bt_{42}&\bt_{41}\end{pmatrix}=\widetilde{\bx}_2^2+\widetilde{\bx}_3^2+q_+\big(\widetilde{\bx}_1^2-\widetilde{\bx}^2_0\big).
\end{gather*}

\subsection{The quantum Klein space}
We want to replicate this construction for the Klein space. This is the most interesting case since, as we already observe, we have 4 different ways to realize it.
\begin{itemize}\itemsep=0pt
 \item We start with the real form $\big(\mathcal{M}_{4}^{k,q},*_{_{Kq}}\big)$; in this case we can work with $\C$ and $\Cs$ at the same time. We introduce the following real variables
\begin{gather*} \bx_{0}=\tfrac 12(\bt_{31}+\bt_{42}), \!\!\qquad \bx_2=\tfrac 12(\bt_{32}+\bt_{41}),\!\!\qquad
\bx_3=\tfrac 1 {2} (\bt_{41}-\bt_{32}), \!\!\qquad \bx_1=\tfrac{1}{2}(\bt_{31}-\bt_{42}),\! \end{gather*}
with commutation relations:
\begin{gather*}
\bx_2\bx_0=q_+\bx_0\bx_2-q_-\bx_0\bx_3,\\
\bx_2\bx_1=q_+\bx_1\bx_2-q_-\bx_1\bx_3,\\
\bx_3\bx_0=q_+\bx_0\bx_3-q_-\bx_0\bx_2,\\
\bx_3\bx_1=q_+\bx_1\bx_3-q_-\bx_1\bx_2,\\
\bx_0\bx_1=\bx_1\bx_0,\\
\bx_2\bx_3=\bx_3\bx_2+q_-\big(\bx_0^2-\bx_1^2\big),
\end{gather*}

The quantum determinant now yields
\begin{gather*}
Q_q=\det\nolimits_q\begin{pmatrix}\bt_{32}&\bt_{31}\\\bt_{42}&\bt_{41}\end{pmatrix}=\bx_2^2-\bx_3^2+q_+\big(\bx_1^2-\bx^2_0\big). 
\end{gather*}
As expected, we get the deformation of a split signature norm.
\end{itemize}

We analyze now the other 2 options one has working within the algebra $\C_s$ (refer to \cite{flm}). In this case we speak of a $\C_s$ quantum deformation.
\begin{itemize}\itemsep=0pt
\item Fixed points of $*_{Mq}$ on $\mathcal{M}_{4}^{\mathbb{\C}_s,q}$ are given by the fixed points
 \begin{gather*}\by_{0}=\tfrac 12(\bt_{31}+\bt_{42}),\!\!\qquad \by_2=\tfrac 12(\bt_{32}+\bt_{41}),\!\!\qquad \by_3=
\tfrac j {2} (\bt_{41}-\bt_{32}), \!\!\qquad \by_1=\tfrac 12(\bt_{31}-\bt_{42}),\end{gather*}
whose commutations relations can be easily obtained by the previous one and the norm turns out to be
\begin{gather*}
Q_q=\det\nolimits_q\begin{pmatrix}\bt_{32}&\bt_{31}\\\bt_{42}&\bt_{41}\end{pmatrix}=\by_2^2-\by_3^2+q_+\big(\by_1^2-\by^2_0\big). 
\end{gather*}
Note that this real form of the complex Minkowski space is isomorphic to the previous one by the simple replacement $\by_3\rightarrow j\by_3$.

\item In the end we can also consider the real form $\big(\mathcal{M}_{4}^{\mathbb{C}_s,q},*_{_{Eq}}\big)$; in this case a convenient choice of fixed points is given by
 \begin{gather*}\tilde{\by}_{0}=\tfrac j2\big(q^{\frac 1 2}\bt_{31}+q^{-\frac 1 2}\bt_{42}\big),\qquad
\tilde{\by}_2=\tfrac 12(\bt_{32}+\bt_{41}),\qquad \tilde{\by}_3= \tfrac j {2} (\bt_{41}-\bt_{32}), \\ \tilde{\by}_1=\tfrac 12\big(q^{\frac 1 2}\bt_{31}-q^{-\frac 1 2}\bt_{42}\big),\end{gather*}
and the metric, as well as the commutation relations coincide, formally, with the previous one
\end{itemize}
We have thus obtained four isomorphic constructions of the $D=4$ quantum Klein space $\mathcal{M}_{2,2}^{q}$. These correspond to the fixed points of the involutions $\ast _{Mq}$, $\ast _{Eq}$ on $\mathcal{M}_{4}^{\mathbb{C}_{s},q}$, and of $\ast _{Kq}$ on $\mathcal{M}_{4}^{k,q}$.

\subsection{The quantum Klein group}\label{qkp-subsec}

In the classical setting we have that the metric $Q$ is preserved by the action of the Klein group; it is natural then to expect that its quantum deformation $Q_q$, introduced in the previous section is preserved under the action of the \textit{quantum Klein group}, viewed as a suitable real form of the quantum Lorentz group $\mathcal{L}_q$. We now work for convenience with $\C_s$ only and we suppress the subscript $k$. We can define three antilinear involutions (i.e., real forms) on the quantum Lorentz group (for the first one, see~\cite{fl}, as for the other two, \cite[pp.~102 and 316]{ks}):
\begin{align*}
\bullet \ *_{PM}\colon \ &
 \mathcal{L}_q  \longrightarrow \mathcal{L}_q ,\\
& x_{ij} \longmapsto S(y_{ji}), \\
& y_{ij} \longmapsto S(x_{ji}), \\
& q \longmapsto q;\\
\bullet \ *_{PE} \colon \ &
\mathcal{L}_q  \longrightarrow \mathcal{L}_q, \\
& x_{ij} \longmapsto S(x_{ji}), \\
& y_{ij} \longmapsto S(y_{ji}), \\
& q \longmapsto q ;\\
\bullet \  *_{PK}\colon \ &
\mathcal{L}_{q}  \longrightarrow \mathcal{L}_{q}, \\
& x_{ij} \longmapsto x_{ij}, \\
& y_{ij} \longmapsto y_{ij}, \\
& q \longmapsto q,
\end{align*}
where $S(x_{ij})$ denotes the antipode for the quantum matrix $x_{ij}$, that is
\begin{gather*}
S\begin{pmatrix} x_{11} & x_{12} \\ x_{21} & x_{22} \end{pmatrix} =
\begin{pmatrix} x_{22} & -qx_{12} \\
-q^{-1}x_{21} & x_{11} \end{pmatrix}
\end{gather*}
and similarly for $S(y_{ij})$. Observe that $ *_{PK}$ is multiplicative and not antimultiplicative; this property could be again restored by sending $q$ to $q^{-1}$. As one can readily check, these maps are well defined and give a $*$-structure on $\mathcal{L}_{q}$.

We have the following proposition, whose proof consists of a tedious direct calculation.

\begin{Proposition}\quad
\begin{enumerate} \itemsep=0pt
\item[$1.$] The quantum Lorentz group acts on the quantum Minkowski space as follows
\begin{align*}
\lambda_{L}\colon \ \mathcal{M}_{4}^{q} & \lra \mathcal{L}_{q}\otimes \mathcal{M}_{4}^q, \\
 \bt_{ij} & \longmapsto \sum_{s,r} y_{is}S(x_{rj}) \otimes \bt_{sr}.
\end{align*}

\item[$2.$] The maps $*_{PA} $, with $A=M,E,K$ define a $*$-structure and a the real form of $\mathcal{L}_{q}$. Furthermore, the Klein space is quantum homogeneous with respect to the coaction of the quantum Klein group. In other words we have
\begin{gather*}
(*_{PA} \times *_{Aq}) \circ \lambda_L=\lambda_L \circ *_{Aq}.
\end{gather*}

\item[$3.$] The quantum Klein quadratic form is a coinvariant for the action of the quantum Lorentz group
\begin{gather*}
\lambda_L(Q_q)=1 \otimes Q_q.
\end{gather*}
\end{enumerate}
\end{Proposition}
\begin{proof} The first point is a direct consequence of Proposition (\ref{coaction-min}).

 The relation $
(*_{PA} \times *_{Aq}) \circ \lambda_L=\lambda_L \circ *_{Aq}
$ arises from a direct calculation; consider for example the element $\bt_{11}$ and the Minkowskian involutions, then one has
\begin{gather*}
\lambda_{L}(\bt_{11})=y_{11}S(x_{11})\otimes\bt_{11}+y_{12}S(x_{21})\otimes\bt_{22}+y_{12}S(x_{11})\otimes\bt_{21}+y_{11}S(x_{21})\otimes\bt_{12}\\
\hphantom{\lambda_{L}(\bt_{11})}{} = y_{11}x_{22}\otimes\bt_{11}-q^{-1}y_{12}x_{21}\otimes\bt_{22}+y_{12}x_{22}\otimes\bt_{21}-q^{-1}y_{11}x_{21}\otimes\bt_{12}.
\end{gather*}
Applying then $(*_{PM} \times *_{Mq})$ and keeping in mind that the Minkowskian involution is antimultiplicative one gets
\begin{gather*}
(*_{PM} \times *_{Mq})\circ \lambda_{L}(\bt_{11})=y_{11}x_{22}\otimes\bt_{11}-q^{-1}y_{12}x_{21}\otimes\bt_{22}-q^{-1}y_{11}x_{21}\otimes\bt_{12}+y_{12}x_{22}\otimes\bt_{21}\\
\hphantom{(*_{PM} \times *_{Mq})\circ \lambda_{L}(\bt_{11})}{} =\lambda_{L}(\bt_{11}),
\end{gather*}
and since $*_{Mq}(\bt_{11})=\bt_{11}$ the result follows naturally. Similarly one proves the same for the other elements and involutions.

To prove the point 3 it is enough to observe that, due to the Manin relation and the parabolic structure we consider, one has $x_{ij}y_{km}=y_{km}x_{ij}$ and after some elementary reordering and algebras, one gets
\begin{gather*}
\lambda_L(Q_q)=\det(y)\det(x)\otimes Q_q.
\end{gather*}
Considering that ${\det}_q(x)=1={\det}_q(y)$, we obtain that the quantum metric is coinvariant as claimed.
\end{proof}

Observe that this proposition tells us in fact that we have a well defined coaction of $(\mathcal{L}_{q},*_{PA})$ on $\big(\mathcal{M}_{4}^{q} ,*_A\big)$ preserving the quantum metric.

\subsection{Formulation with the algebraic star product}
In this section, we want to present the results obtained in the previous section in terms of classical objects. In particular we want to take advantage of the isomorphism between the algebra $\textrm{M}_q(2)$ and $\big(\mathcal{O}(\mathcal{M}_{2,2})\big[q,q^{-1}\big],\star\big)\equiv k_q[\tau_{41},\tau_{42},\tau_{31},\tau_{32}]$, where $\star$ is certain non-commutative product naturally inherited from the underlying quantum group structure. This approach and language is probably more familiar to physicists because of its versatility. As mentioned before, the starting point for this construction is the observation that the map
\begin{align*}
\Phi_q\colon \ k_q[\tau_{41},\tau_{42},\tau_{31},\tau_{32}]&\longrightarrow \textrm{M}_q(2),\\
 \tau_{41}^m \tau_{42}^n \tau_{31}^p \tau_{32}^r&\mapsto \bt_{41}^m \bt_{42}^n \bt_{31}^p \bt_{32}^r
\end{align*}
is a module isomorphism and thus has an inverse; $\Phi_q:$ is called quantization map, and essentially it encodes a choice of ordering. It is here worth stressing that with this choice of ordering one gets a basis for $\textrm{M}_q(2)$ as proved in \cite{Cerv-1}, but in principle one could make other compatible choices. Next, we define the following non commutative product
\begin{gather*}
f \star g=\Phi_q^{-1}(\Phi_q(f) \Phi_q(g)) , \qquad f\in \big(\mathcal{O}(\mathcal{M}_{2,2})\big[q,q^{-1}\big]\big).
\end{gather*}
Using Manin's relations (\ref{manin-rel2}), then one can easily construct the following
\begin{gather*}
\big(\tau_{41}^a\tau_{42}^b\tau_{31}^c\tau_{32}^d\big)\star \big(\tau_{41}^m\tau_{42}^n\tau_{31}^p\tau_{32}^r\big) = q^{-m(c+b)-d(n+p)}\big(\tau_{41}^{a+m} \tau_{42}^{b+n} \tau_{31}^{c+p} \tau_{32}^{d+r}\big)\\
\qquad{}+ \sum\limits_{k=1}^{\min (d,m)} q^{(k-m)(c+b)+(k-d)(n+p)}F(k,q,d,m) \tau_{41}^{a+m-k}\tau_{42}^{b+k+n}\tau_{31}^{c+k+p}\tau_{32}^{d-k+r},
\end{gather*}
where $F(k,q,d,m)$ is a function defined recursively (see \cite{Cerv-1} for more detail). We comment that there exists a (unique) differential star product, thus acting on $C^{\infty}$ functions, that coincides with the one given above on polynomials.

We can now pullback the metric to the star product algebra, obtaining
\begin{gather*}
Q=\Phi_q^{-1}(Q_q)=x_2^2+x_3^2-qx_0^2-qx_1^2,
\end{gather*}
where the star product among polynomials on the real variables $x$ and is naturally inherited from the $\bx$ commutation relations.

\section[The $\mathcal{N}=1$ quantum Klein superspace $\mathcal{M}_{2,2|1}^{q}$]{The $\boldsymbol{\mathcal{N}=1}$ quantum Klein superspace $\boldsymbol{\mathcal{M}_{2,2|1}^{q}}$}\label{smink-sec}

The extension of the natural construction that we have done in the previous sections to the supersetting offers no difficulty, hence we will summarize quickly the relevant definitions and the results. For all of the supergeometry terminology we refer the reader to~\cite{ccf, fl}.

\subsection{Quantum supergroups} \label{qsgrps-subsec}

We start with the definition of quantum matrix superalgebra and quantum special linear supergroup. Let $k=\C$ or $\C_s$, as in beginning of Section~\ref{prelim-sec} and $k_q=k\big[q,q^{-1}\big]$.

\begin{Definition} We define \textit{quantum matrix superalgebra}
\begin{gather*}
\rM_q(m|n)\overset{\rm def}{=}k_q\langle a_{ij} \rangle /I_M,
\end{gather*}
where $I_M$ is generated by the relations \cite{ma2}
\begin{gather*}
a_{ij}a_{il}=(-1)^{\pi(a_{ij})\pi(a_{il})} q^{(-1)^{p(i)+1}}a_{il}a_{ij}, \qquad j < l,\nonumber \\
a_{ij}a_{kj}=(-1)^{\pi(a_{ij})\pi(a_{kj})} q^{(-1)^{p(j)+1}}a_{kj}a_{ij}, \qquad i < k, \nonumber\\
a_{ij}a_{kl}=(-1)^{\pi(a_{ij})\pi(a_{kl})}a_{kl}a_{ij}, \qquad i< k,j > l \quad \text{or} \quad i > k,j < l,\nonumber \\
a_{ij}a_{kl}-(-1)^{\pi(a_{ij})\pi(a_{kl})}a_{kl}a_{ij}= \eta\big(q^{-1}-q\big)a_{kj}a_{il}, \qquad i<k,j<l,
\end{gather*}
where
\begin{gather*}
\eta=(-1)^{p(k)p(l)+p(j)p(l)+p(k)p(j)}
\end{gather*}
with
$p(i)=0$ if $1 \leq i \leq m$, $p(i)=1$ otherwise and $\pi(a_{ij})=p(i)+p(j)$ denotes the parity of $a_{ij}$. $\rM_q(m|n)$ is a bialgebra with the usual comultiplication and counit
\begin{gather*}
\Delta(a_{ij})=\sum a_{ik} \otimes a_{kj},\qquad \ep(a_{ij})=\de_{ij}.
\end{gather*}

We define the \textit{special linear quantum supergroup} $\rSL_q(m|n)$ as $\rM_q(m|n)/(B_q-1)$, where $B_q$ is the \textit{quantum Berezinian}, which is a central element in $\rM_q(m|n)$ (see \cite{fi5,P, zhang} and \cite[Chapter~5, Section~5.4]{fl} for more details).
\end{Definition}

We now turn to the more relevant definitions for us, namely the quantum Poincar\'{e} super\-group~$\mathcal{SP}_{q}$ and the $4|1$ dimensional quantum Minkowski superspace $\mathcal{M}^{k,q}$ as quantum homogeneous space, that is together with a coaction of $\mathcal{SP}_{q}$ on it.

\begin{Definition}We define {\it quantum Poincar\'{e} supergroup} $\mathcal{SP}_{q}$ as the quotient of $\rSL_q(m|n)$ by the generators $a_{ij}$, $\alpha_{k5}$, $\alpha_{5l}$, $i,j=1,2$ or $i,j=3,4$ or $i=3,4$, $j=1,2$ and $k=1,2$, $l=3,4$. In quantum matrix form
\begin{gather*}
\mathcal{SP}_q=\left(\begin{matrix}
L & 0 & 0 \\ M & R & \chi \\ \phi & 0 & d
\end{matrix}\right),
\end{gather*}
where
\begin{gather*}
L=\left(\begin{matrix} a_{11} & a_{12} \\ a_{21} & a_{22}
\end{matrix}\right), \qquad
M=\left(\begin{matrix} a_{31} & a_{32} \\ a_{41} & a_{42}
\end{matrix}\right), \qquad
R=\left(\begin{matrix} a_{33} & a_{34} \\ a_{43} & a_{44}
\end{matrix}\right), \\
\chi=\left(\begin{matrix} \al_{35} \\ \al_{45}
\end{matrix}\right), \qquad
\phi=\left(\begin{matrix} \al_{53} & \al_{54}
\end{matrix}\right), \qquad
d=a_{55}
\end{gather*}
in terms of the generators $a_{ij}$, $\al_{kl}$ of $\rSL_q(m|n)$, the quantum special linear supergroup.
\end{Definition}

\subsection{Quantum chiral Minkowski superspace}\label{qsmink-subsec}

We now define the {\it quantum Minkowski superspace} $\mathcal{SM}^{k,q}$ as generated by the matrices $t$ and $\tau$, in analogy with our heuristic derivation in~(\ref{eur-der}), obtained through the equality
\begin{gather*} 
\begin{pmatrix} L & 0 & 0\\ tL & R & 0 \\ \tau L & \nu & d
\end{pmatrix}
\begin{pmatrix}I_2 & s & \sigma \\ 0 & I_2 & \rho \\
0 & 0 & 1\end{pmatrix}
=\begin{pmatrix}a_{11} & \dots & a_{14} & \al_{15} \\
 \vdots & & \vdots & \vdots\\
a_{41} & \dots & a_{44} & \al_{45}\\
 \al_{51} & \dots & \al_{54}& a_{55}\end{pmatrix}.
\end{gather*}

After a small calculation, the following result is achieved:
\begin{gather*}
t=(\bt_{kj})=\begin{pmatrix}-q^{-1}D_{23} {D_{12} }^{-1} &
D_{13} {D_{12} }^{-1} \\
-q^{-1}D_{24} {D_{12} }^{-1} &
D_{14} {D_{12} }^{-1} \end{pmatrix},\nonumber\\
\tau=(\tau_{5j})=\begin{pmatrix}-q^{-1}D_{25} {D_{12} }^{-1} &
D_{15} {D_{12} }^{-1} \end{pmatrix}.
\end{gather*}
As above, both $t$ and $\tau$ are quantum matrices, once the first and second columns are interchanged.

\begin{Remark} Notice that our definition is over $k$ and so works for $k=\C$, giving the quantum complex Minkowski superspace, but also over $k=\C_s$, hence giving the quantum $\C_s$ Minkowski superspace. The advantage of our unified treatment is that, when looking at real forms, we will obtain at once both the real Minkowski and real Klein quantum superspaces.
\end{Remark}

Our definition corresponds to identify the Minkowski superspace with the translation supergroup inside the Poincar\'{e} supergroup as we did in Section~\ref{qkm-subsec}.

\begin{Proposition}\label{supercoaction}
The $\mathcal{N}=1$ quantum Minkowski superspace $\mathcal{{S}M}^{k,q}$ is a quantum homogeneous superspace for the quantum Poincar\'{e} supergroup. The coaction is explicitly given as
\begin{align*}
\delta^{(s)} \colon \ \mathcal{M}^{k,q} & \lra \mathcal{SP}_q \otimes \mathcal{M}^{k,q}, \\
\bt_{ij} & \mapsto t_{ij} \otimes 1+\sum_{u,v} r_{iu}S(\ell_{vj})\otimes \bt_{uv} + \sum_{v}\chi_{i5} S(\ell_{vj})\otimes \tau_{5v}, \\
\tau_{5j} & \mapsto \psi_{5j} \otimes 1 +\sum_v d S(\ell_{vj}) \otimes \tau_{5v},
\end{align*}
where $R=(r_{ij})$, $L=(\ell_{kl})$. $($To ease the notation, we replace the undetermined in $M$ with $NL$, for a suitable $N$ and similarly $\phi$ with $\psi L$ for a suitable $\psi.)$
\end{Proposition}

\begin{proof} This again follows naturally from our construction since the coaction of $\delta^{(s)}$ corresponds to the restriction of the $\rSL_q(4|1)$ coaction to the blocks $\bt$ and $\tau$.
 \end{proof}

\subsection{Real forms}\label{superreal-subsec}

The chiral complex Minkowski superspace does not admit a physically interesting real form, because the odd part is spinorial and in this case, the (Weyl) semispinors are complex. In the unquantized setting, this defect is fixed by considering extra odd coordinates and pairing them to obtain such a real form (see \cite[Chapter~4]{fl} for an exhaustive treatment and the references within). A quantization of the quantum real Minkowski superspace presents then difficulties, since the extra odd coordinates must also be quantized and calculations become intricated.

On the other hand, the quantization of the chiral Klein superspace follows very naturally in our construction. It should be here recalled that in the non supersymmetric case we have defined three different involution yielding the same geometrical structure; in the following, we will focus for simplicity on the $\mathcal{N}=1$ superextension of the case associated to the Klein involution.

We define the \textit{quantum real chiral Klein superspace} as the pair $\big(\mathcal{SM}^{k,q},{*_{SKq}}\big)$, where $*_{SKq}$ is the antilinear involution
\begin{align*}
*_{SKq}\colon \qquad \mathcal{SM}^{k,q} &\lra \mathcal{{S}M}^{k,q}, \\
\begin{pmatrix}\bt_{31} & \bt_{32} \\ \bt_{41} & \bt_{42} \end{pmatrix} &\longmapsto \begin{pmatrix}\bt_{31} & \bt_{32} \\ \bt_{41} & \bt_{42}\end{pmatrix}, \\
\begin{pmatrix}\tau_{51} & \tau_{52} \end{pmatrix}&\longmapsto \begin{pmatrix}\tau_{51} & \tau_{52}\end{pmatrix},\\
{q}&\longmapsto {q},
\end{align*}
which is the identity on the coordinates. Note that since we want it to be multiplicative as in the even case, we need to send $q$ to $q$.

Our construction is compatible to what we have discussed in our previous sections. The metrics are the ones discussed in Section~\ref{qmin-subsec}, since the odd variables do not modify them (see also \cite[Chapter~4]{fl}).

$\big(\mathcal{SM}^{kq},{*_{SKq}}\big)$ is a real homogeneous quantum superspace, when the antilinear multiplicative involution on the quantum supergroup $\mathcal{SP}_q$ (which is the identity on the generators and send $q$ to itself) is considered. One can in fact immediately see that preserves the coaction as given in Proposition~\ref{supercoaction}.

\section{Conclusions}
In this paper we propose an algebraic quantization, in the sense of quantum groups, of the complex and split complex Minkowski spaces $\mathcal{M}_{4}^{\mathbb{C},q}$ and $\mathcal{M}_{4}^{\mathbb{C}_s,q}$ viewed as quantum homogenous spaces; we focus our attention on their real forms, yielding the Lorentzian, Kleinian and Euclidean signatures $(3,1)$, $(2,2)$, $(4,0)$. All this is summarized in the following table:
\begin{center}
\begin{tabular}{c|l}
\hline
real quantum space&signature\\[1mm] \hline
$(*_{_{Kq}},\mathcal{M}_{4}^{\C,q}$) &$(2,2)$\\
 $(*_{_{Kq}}\mathcal{M}_{4}^{\C_s,q})$ &$(2,2)$\\
\hline
$(*_{_{Mq}},\mathcal{M}_{4}^{\C,q}$) &$(3,1)$\\
$(*_{_{Mq}},\mathcal{M}_{4}^{\C_s,q})$ &$(2,2)$\\
\hline
 ($*_{_{Eq}},\mathcal{M}_{4}^{\C,q})$ &(4,0)\\
$(*_{_{Eq}},\mathcal{M}_{4}^{\C_s,q}$) &(2,2)
 \end{tabular}
\end{center}
The beauty of our approach is that those spaces are naturally endowed with the (co)action of the corresponding isometry quantum group, as we explicitly show; moreover, for all these spaces, we give an explicit representation of the deformed coordinates algebra. In this setting we also extend such analysis to the supercase by constructing the chiral Klein superspace.

\subsection*{Acknowledgments}

We would like to thank Professors Francesco Bonechi, Meng-Kiat Chuah and Fabio Gavarini for useful discussions and helpful comments. We also wish to thank our anonymous referees for helpful comments, which have helped us to improve the clarity of our paper. A.M.~wishes to thank the Department of Mathematics at the University of Bologna, for the kind hospitality during the realization of this work.

\pdfbookmark[1]{References}{ref}
\LastPageEnding

\end{document}